\documentclass[10pt]{article}
\pdfoutput=1

\usepackage{anysize} 
\marginsize{2.4cm}{2.4cm}{1cm}{2cm}
\usepackage{fancyhdr} 
\usepackage{amsmath} 
\usepackage{amssymb} 
\usepackage{amsthm} 
\usepackage{graphicx} 
\usepackage{booktabs} 
\usepackage{complexity} 
\usepackage[numbers]{natbib} 
\usepackage{color} 
\usepackage{colortbl} 
\usepackage[colorinlistoftodos,textwidth=4cm]{todonotes} 
\usepackage{setspace} 
\usepackage{minitoc} 
\usepackage{enumitem} 
\usepackage{floatrow} 
\usepackage{shorttoc} 
\usepackage{algorithmic} 
\usepackage{array} 
\usepackage{tikz} 
\usetikzlibrary{shapes,arrows,topaths}
\usepackage{float}
\usepackage[pdftex,pagebackref=false,colorlinks]{hyperref} 
\hypersetup{linkcolor=blue,filecolor=blue,citecolor=blue,urlcolor=blue}
\usepackage{algorithm} 
\usepackage{tabularx} 
\usepackage{caption} 
\usepackage{subfig} 
\usepackage{pdfsync}

\theoremstyle{plain}
\newtheorem{theorem}{Theorem}[section]

\theoremstyle{definition}

\theoremstyle{remark}

\newcommand{\namedref}[2]{\hyperref[#2]{#1~\ref*{#2}}}

\newcommand{\sectionref}[1]{\namedref{Section}{#1}}

\newcommand{\figureref}[1]{\namedref{Figure}{#1}}
\newcommand{\tableref}[1]{\namedref{Table}{#1}}

\newcommand{\itemref}[1]{\namedref{Item}{#1}}

\newcommand{\algorithmref}[1]{\namedref{Algorithm}{#1}}

\newcommand{\stepref}[1]{\namedref{Step}{#1}}

\newcommand{\textred}[1]{\textcolor{red}{#1}}
\ifx\noeditingmarks\undefined
  \newcommand{\pgwrapper}[2]{\textred{#1: #2}}
\else
  \newcommand{\pgwrapper}[2]{}
\fi

\newcommand{\parhead}[1]{\medskip \noindent {\bfseries\boldmath\ignorespaces #1.}\hskip 0.9em plus 0.3em minus 0.3em}

\newcommand{\CHILD}{\mathsf{C}}
\newcommand{\ch}{\CHILD}
\newcommand{\PARENT}{\mathsf{P}}
\newcommand{\pa}{\PARENT}
\newcommand{\DESTINATION}{\mathsf{D}}
\newcommand{\PARENTSET}{\mathcal{P}}
\newcommand{\REQUIREDSET}{\mathcal{R}}

\newcommand{\NODE}{\mathsf{N}}
\newcommand{\no}{\NODE}
\newcommand{\SOURCE}{\mathsf{S}}
\newcommand{\so}{\SOURCE}

\newcommand{\pk}{\mathsf{pk}}
\newcommand{\sk}{\mathsf{sk}}
\newcommand{\INDEG}[1]{|\REQUIREDSET_{#1}|}
\newcommand{\one}{PIP}
\newcommand{\two}{Log-PIP}

\newcommand{\cert}{\mathsf{cert}}

\newcommand{\pollutionwork}{\cite{KrohnFM04}, \cite{CharlesJL06}, \cite{GkantsidisR06}, \cite{GkantsidisR06}, \cite{ZhaoKMH07}, \cite{YuWRG08}, \cite{HoLKMEK08}, \cite{ JaggiLKHKME08}, \cite{BonehFKW09}, \cite{KosutTT09}, \cite{AgrawalB09},  \cite{DongCN09}, \cite{AgrawalBBF10}, \cite{LeM10}, \cite{YaoSJL10}, and \cite{WanVNK10}}

\newcommand{\pollutiondetection}{\cite{KrohnFM04}, \cite{CharlesJL06}, \cite{GkantsidisR06}, \cite{GkantsidisR06}, \cite{ZhaoKMH07}, \cite{YuWRG08}, \cite{BonehFKW09}, \cite{DongCN09}, \cite{AgrawalBBF10}, and \cite{WanVNK10}}

\newcommand{\pollutioncrypto}{\cite{BonehFKW09}, \cite{ZhaoKMH07}, \cite{KrohnFM04}, \cite{YuWRG08}, \cite{CharlesJL06}, \cite{AgrawalB09}, and \cite{AgrawalBBF10}}

\newcommand{\pollutionsignatures}{\cite{BonehFKW09}, \cite{ZhaoKMH07}, \cite{KrohnFM04}, \cite{YuWRG08}, \cite{CharlesJL06}, and \cite{AgrawalBBF10}}

\newcommand{\gen}{\mathsf{gen}}
\newcommand{\sig}{\mathsf{sig}}
\newcommand{\ver}{\mathsf{ver}}

\newcommand{\checkh}{\textsc{CheckHelper}}
\newcommand{\vt}{\textsc{VerifTest}}
\newcommand{\combine}{\textsc{Combine}}

\newcommand{\oneabs}{3.7}
\newcommand{\twoabs}{1.4}
\newcommand{\onerel}{2}
\newcommand{\tworel}{0.5}

\title{Going Beyond Pollution Attacks: \\ Forcing Byzantine Clients to Code Correctly}
\date{July 29, 2011}
\author{
Raluca Ada Popa\thanks{Email: \href{mailto:ralucap@mit.edu}{raluca@csail.mit.edu}.} \\ MIT CSAIL
\and
Alessandro Chiesa\thanks{Email: \href{mailto:alexch@mit.edu}{alexch@csail.mit.edu}.} \\ MIT CSAIL
\and
Tural Badirkhanli\thanks{Email: \href{mailto:turalb@mit.edu}{turalb@csail.mit.edu}.} \\ MIT CSAIL
\and
Muriel M\'{e}dard\thanks{Email: \href{mailto:medard@mit.edu}{medard@mit.edu}.} \\ MIT RLE
}

\begin{document}

\maketitle

\begin{abstract}

Network coding achieves optimal throughput in multicast networks. However, throughput optimality \emph{relies} on the network nodes or routers to code \emph{correctly}. A Byzantine node may introduce junk packets in the network (thus polluting downstream packets and causing the sinks to receive the wrong data) or may choose coding coefficients in a way that significantly reduces the throughput of the network.

Most prior work focused on the problem of Byzantine nodes polluting packets. However, even if a Byzantine node does not pollute packets, he can still affect significantly the throughput of the network by not coding correctly. No previous work attempted to verify if a certain node \emph{coded correctly using random coefficients} over \emph{all} of the packets he was supposed to code over.

We provide two novel protocols (which we call \one{} and \two{}) for detecting whether a node coded correctly over all the packets received (i.e., according to a random linear network coding algorithm). Our protocols enable any node in the network to examine a packet received from another node by running a ``verification test''. With our protocols, the worst an adversary can do and still pass the packet verification test is in fact equivalent to random linear network coding, which has been shown to be optimal in multicast networks. Our protocols resist collusion among nodes and are applicable to a variety of settings.

Our topology simulations show that the throughput in the worst case for our protocol is two to three times larger than the throughput in various adversarial strategies allowed by prior work. We implemented our protocols in C/C++ and Java, as well as incorporated them on the Android platform (Nexus One). Our evaluation shows that our protocols impose modest overhead. 

\end{abstract}

\newpage
\tableofcontents
\newpage

\section{Introduction}
\label{sec:introduction}

Network coding was first proposed by Ahlswede et al. \cite{AhlswedeCLY00}, who demonstrated that, for certain networks, network coding can produce a higher throughput than the best routing strategy. A subsequent line of work that includes the works of Koetter et al. \cite{KoetterM03}, Li et al. \cite{LiYC03}, and Jaggi et al. \cite{JaggiSCEEJT05} showed that random linear coding reaches maximum throughput for multicast networks. Overall, network coding has proved better than routing for both wired and wireless networks and for both multicast and broadcast \cite{NarmawalaS08}; it has also found applications to increasing the robustness and throughput of peer-to-peer networks (e.g., \cite{GkantsidisR05}) and to a variety of sensor wireless networks as surveyed by Narmawala and Srivastava \cite{NarmawalaS08}.

\parhead{Throughput optimality requires diversity}
The throughput guarantees of network coding, however, \emph{rely on the assumption that all the nodes in the network code correctly}, i.e., each node in the network, when receiving packets, is assumed to transmit a packet that is a random linear combination of the incoming packets; informally, packets that are indeed linear combinations of the incoming packets are said to be \emph{valid}, and packets that are \emph{random} linear combinations of the incoming packets are said to be \emph{diverse}.

The assumption that each node in the network codes correctly may not hold because the network may contain \emph{Byzantine nodes}, who are malicious or faulty nodes.
For example, a Byzantine node may change the payload or the coding vector in a way that is not a linear combination of the received packets, thereby transmitting an \emph{invalid} (or \emph{polluted}) packet. The invalid packet will mix with other packets and thus pollute more packets, ultimately causing the decoded information at the sinks to be incorrect.

In fact, a Byzantine node can transmit a valid packet (i.e., a linear combination of the received packets), but still manage to decrease the overall throughput at the sinks. The Byzantine node could choose coefficients for the linear combination in a way that is \emph{not random}: the node could forward one of the packets (by simply routing), code over only a subset of the packets, or, even worse, choose coefficients that do not contribute any new information to his receivers, thus, \emph{effectively sending nothing}. While the network is not polluted by such a Byzantine node (and the decoded information at the sinks is still valid), the throughput of the network is decreased. In \sectionref{sec:evaluation}, as an example, we show that such Byzantine nodes can indeed reduce the throughput to \emph{as much as a half or a third} in some specific cases and as much as $20\%$ on random topologies. \figureref{fig:butterfly} shows a simple example of $50\%$  throughput reduction on the standard butterfly topology.

\begin{figure}[h!]
\centering
\subfloat[]{\includegraphics[scale=0.60]{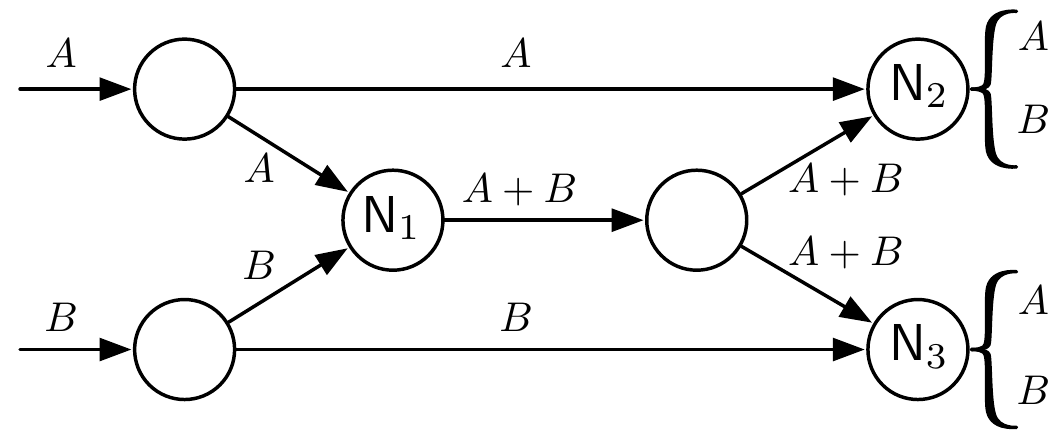}\label{fig:butterfly-a}}
\qquad
\subfloat[]{\includegraphics[scale=0.60]{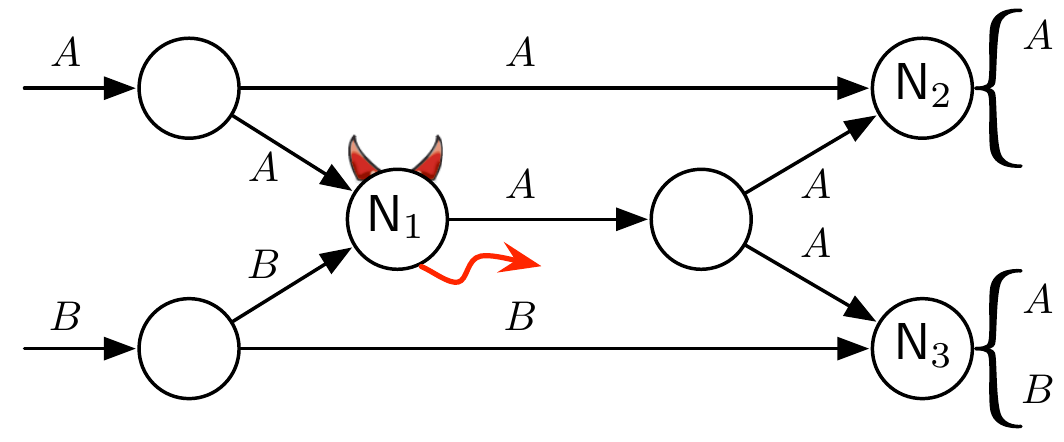}\label{fig:butterfly-b}}
\captionsetup{width=0.85\textwidth}
\caption{Example of throughput reduction caused by a Byzantine node on a butterfly network: (a) if node $\NODE_{1}$ is honest, he will send $A+B$, thus allowing both $\NODE_{2}$ and $\NODE_{3}$ to recover both $A$ and $B$; (b) if node $\NODE_{1}$ is Byzantine, he may choose to send only $A$, thus halving the throughput at $\NODE_{2}$, which can now recover only $A$.}
\label{fig:butterfly}
\end{figure}

\parhead{Insufficiency of prior work to guarantee correctness}
A significant body of previous work that includes \pollutionwork{} addressed the problem of defending against pollution attacks, where the goal is to enforce or check that the packets sent by each node to be \emph{some} (not necessarily random) linear combination of the packets sent by the source. Most prior work on enforcing validity of packets has focused on detecting polluted packets right at the point where a Byzantine node injected them into the network \pollutionsignatures: when a Byzantine node injects an invalid packet into the network, the node receiving the packet is able to detect if the packet is invalid by running a test, and can discard the invalid packet right away. 

However, all such work does not detect Byzantine nodes that deviate from \emph{random} linear coding of the received packets, thus allowing such Byzantine nodes to reduce throughput as already discussed above. In particular, Byzantine nodes are still allowed to simply forward a received packet (rather than to code over multiple packets) or use coefficients that provide no new degrees of freedom to downstream nodes, effectively sending no data.

\parhead{Our result}
Given that Byzantine nodes may significantly affect the throughput of the network, we believe that it is important to study the following problem:
\begin{center}
How to force a node to code \emph{correctly}? \\
(That is, to code both \emph{validly} and \emph{randomly}, over \emph{all} the received packets.)
\end{center}
Our main contribution is a novel protocol for enabling each child of a node to detect whether the node coded correctly over all the packets he was supposed to code over (i.e., according to a random linear network coding algorithm).  In our protocol, a child node $\CHILD$ of a node $\NODE$ (where \emph{child} means that $\CHILD$ receives data from his \emph{parent} $\NODE$) can check, by running a \emph{verification test}, that the data from $\NODE$ is the result of correctly coding over the packets $\NODE$ receives from his parents. The node $\CHILD$ need only examine the packet received from $\NODE$ and does not need to know the precise packet payloads used in coding at $\NODE$. 

Let the \textit{required set} of $\NODE$, denoted $\REQUIREDSET_{\NODE}$, be the subset of the parents of $\NODE$ that $\NODE$ is expected to code over. As we will discuss in \sectionref{sec:applications}, the exact definition of the required set depends on the application; the flexibility in defining it will enable our protocols to be applicable to a variety of settings. For example, some applications may require a node to code over the packets from all his parents; other applications, perhaps due to unreliability of the communication channel, may require nodes to code over at least some minimum number of parents. 

Using our protocols presented in \sectionref{sec:protocol}, the child node $\CHILD$ can ensure that:
\begin{itemize}
  \item[(i)] the packet from $\NODE$ is the result of coding over the packets from all the nodes in $\REQUIREDSET_{\NODE}$, and
  \item[(ii)] the coding coefficients used by $\NODE$ are pseudorandom.
\end{itemize}
We provide two algorithms, with two different kinds of guarantees: Payload-Independent-Protocol (\one{}) and Log-Verification \one{} (\two{}). \one{} always detects if $\NODE$ failed to code over all the packets from parents in the required set, whereas \two{} detects such a violation with an adjustable probability. In cases where nodes can have many parents (say, more than $10$), \two{} is faster and more bandwidth efficient. While we use \textit{pseudorandom} coefficients instead of random ones, this does not affect the throughput guarantees of network coding (see \sectionref{sec:force-diverse}); accordingly, we will use the two terms interchangeably in this paper. 

Furthermore, our protocols are \emph{resistant to collusion among nodes}: even if the two Byzantine nodes $\NODE$ and $\CHILD$ collude, the other honest children of $\NODE$ can still check whether $\NODE$ coded correctly over any non-colluding parents. 

Finally, we assume that there exist penalties for nodes that are found to send incorrect packets, and we assume that they drive incentives against cheating in a detectable manner. A discussion of the exact form of such penalties  of course lies outside of the scope of this paper, and one should choose the penalty that is best fit for one's application. To facilitate the use of a penalty system, though, our protocol enables nodes to \emph{prove} (and not only detect) that a parent cheated (i.e., did not code correctly); moreover, Byzantine nodes cannot falsely accuse honest nodes of not coding correctly.

Thus, we assume that Byzantine nodes will not cheat in a detectable way. We therefore consider an adversarial model in which \emph{Byzantine nodes perform the worst possible action to pass the verification test}. In \sectionref{sec:proofs}, we prove that the worst an adversary can do and still pass our packet verification tests is \emph{to code correctly} (i.e., according to a random linear network coding scheme), which has been shown to give optimal throughput in multicast networks. 

\parhead{Implementation and evaluation}
Our simulations in \sectionref{sec:evaluation} show that the throughput in the best adversarial strategy for our protocol is two to three times larger than the throughput in several adversarial strategies allowed by prior work. 

We implemented our protocols in C/C++. We also wrote a Java implementation for Java-based P2P applications and an Android package for smartphone P2P file sharing. Our C/C++ evaluations show that the protocols are reasonably efficient: the running time at a node to prepare for transmitting the data is less than $0.3$ ms, and the time to perform a verification test is $\oneabs$ ms with \one{} and $\twoabs$ ms with \two{}. Compared to the overhead introduced by a pollution detection scheme that we analyzed \cite{BonehFKW09}, the additional overheads introduced by our two protocols are respectively less than $\onerel\%$ for \one{} and less than $\tworel\%$ for \two{}. This suggests that, if one is already using a pollution detection scheme, then additionally enforcing diversity of packets will not affect performance by much. Moreover, the overhead of both of our protocols is independent of how large the packet payload is.

\section{Related Work}
\label{sec:related-work}

Ahlswede et al. \cite{AhlswedeCLY00} have pioneered the field of network coding. They showed the value of coding at routers and provided theoretical bounds on the capacity of such networks. Works such as those of Koetter et al. \cite{KoetterM03}, Li et al. \cite{LiYC03}, and Jaggi et al. \cite{JaggiSCEEJT05} show that, for multicast traffic, linear codes achieve maximum throughput, while coding and decoding can be done in polynomial time. Ho et al. \cite{HoKMKE03} show that random network coding can also achieve maximum network capacity. Network coding has been shown to improve throughput in a variety of networks: wireless \cite{LunMK05}, peer-to-peer content distribution \cite{GkantsidisR05}, energy \cite{WieselthierNE00}, distributed storage \cite{Jia06}, and others.

Despite its throughput benefits, however, network coding is susceptible to Byzantine attacks. A Byzantine node can inject into the network junk packets, which will mix with correct packets and generate more junk packets, thus resulting in junk data at the sink. 

A significant amount of research aims to prevent against or recover from pollution attacks \pollutionwork.
Ho et al. \cite{HoLKMEK08} attempt to detect at the sinks if the packets have been modified by a Byzantine node.  They do so by adding hash symbols that are obtained as a polynomial function of the data symbols, and pollution is indicated by an inconsistency between the packets and the hashes. 

Jaggi et al. \cite{JaggiLKHKME08}, for example, discuss rate-optimal protocols that survive Byzantine attacks. Their idea is to append extra parity information to the source messages. 
Kosut et al.~\cite{KosutTT09} provide non-linear protocols for achieving capacity in the presence of Byzantine adversaries. 





There has also been important work in the problem of detecting polluted packets when they are injected, see for example \pollutiondetection.
These schemes are helpful because they prevent polluted packets from mixing with other packets. The most common approach has been the use of a homomorphic cryptographic scheme (such as signature) \pollutioncrypto.
In a peer-to-peer setting, Krohn et al. \cite{KrohnFM04} propose a scheme based on homomorphic hashes to detect on the fly whether a received packet is valid. The homomorphic hashes are used to verify if the check blocks of downloaded files are indeed a linear combinations of the original file blocks.
 Gkantsidis and Rodriguez \cite{GkantsidisR06} further extend the approach of Krohn et al. to resist pollution attacks in peer-to-peer file distribution systems that use network coding. 
They also mention the entropy attack, which is similar to our diversity attack. However, they do not solve the problem of enforcing a Byzantine client to code diversely. Their approach is to have a node download coding coefficients from neighbors and decide from which neighbors to download the data to get the most innovative packets. However, a Byzantine client can still not code diversely, and for example, can choose not to code over the data from a parent that he knows would provide innovative information to his neighbors, thus reducing overall throughput.

Wan et. al~\cite{WanVNK10} propose limiting pollution attacks by identifying the malicious nodes, so that they can be isolated, and Le and Markopoulou~\cite{LeM10} by identifying the precise location of Byzantine attackers using a homomorphic MAC scheme. 

Zhao et al. \cite{ZhaoKMH07} provides a signatures scheme for content distribution with network coding based on linear algebra and cryptography. The source provides all nodes with an invariant vector and public key information. With that information, all nodes can check on the fly the validity of a packet. \cite{YuWRG08} provides  homormorphic signatures schemes for preventing such Byzantine attacks, but the paper is vacuous due to a flaw. \cite{CharlesJL06} and \cite{BonehFKW09} also provide homomorphic signatures schemes, with a construction based on elliptic curves. This scheme augments the packet size by only one constant of about $1024$ bits. 

Another recent approach to detecting polluted packets is the algebraic watchdog~\cite{KimMB10, LiangAV10} in which nodes sniff on packets from other nodes and try to establish if they are polluted. 

However, while all these schemes only check if a packet is valid, they \emph{cannot establish if a packet is diverse}. If Byzantine nodes are prevented from sending junk packets, because there are packet validity checks, it is still the case that there are other ways in which a Byzantine node can affect the throughput without violating any validity checks. For example, a Byzantine node can simply not send any data, he can forward one of the received packets (without coding), he can code with fixed coefficients, or he can choose coefficients that minimize the network throughput. In \sectionref{sec:evaluation}, we show that Byzantine behavior of this kind does indeed significantly decrease throughput. All these behaviors are not considered (and not prevented) by all previous work on pollution attacks.

\section{Model}
\label{sec:model}

We present the network model and then formulate the security problem that we want to solve. In \sectionref{sec:applications}, we explain how our model and protocols apply to a variety of problem domains.

\subsection{Network Model}
\label{sec:network-model}

We consider a network where nodes perform random linear network coding \cite{HoKMKE03} over some finite field. Roughly, each packet is a pair consisting of a \emph{payload} $M$ and a \emph{coding vector} $C$; nodes ``code'' by choosing random coefficients and using them to compute linear combinations of the received packets. For example: node $\NODE$ receives two packets $\langle M_{1}, C_{1} \rangle$ and $\langle M_{2}, C_{2} \rangle$; to random linear network code these packets, $\NODE$ chooses two random coefficients $\alpha_{1}$ and $\alpha_{2}$ from a certain finite field and computes the resulting coded packet as $\langle \alpha_{1} M_{1} + \alpha_{2} M_{2}, \alpha_{1} C_{1} + \alpha_{2} C_{2} \rangle$, where the computations are also performed in the finite field. In \sectionref{sec:notation}, we provide more details about the structure of a packet.

The network is modeled as a directed graph in the natural way: each node in the network corresponds to a vertex in the graph, and if a node $\NODE$ sends data to another node $\NODE'$ then there is a directed edge in the graph from the vertex (corresponding to the node) $\NODE$ to the vertex (corresponding to the node) $\NODE'$; we then say that $\NODE$ is a \emph{parent} of $\NODE'$ and that $\NODE'$ is a \emph{child} of $\NODE$; similarly, if there is a directed edge from $\NODE'$ to $\NODE''$, we say that $\NODE$ is a \emph{grandparent} of $\NODE''$ and $\NODE''$ is a \emph{grandchild} of $\NODE$. Each node sends one packet per time period to each of his children.

We always denote a generic node in the network by $\NODE$; he has parents denoted by $\PARENT_{i}$ and children denoted by $\CHILD_{j}$. We denote by $\PARENTSET_{\NODE}$ the set of parents of $\NODE$. As discussed, the \emph{required set} of $\NODE$, denoted by $\REQUIREDSET_{\NODE}$, is the subset of $\PARENTSET_{\NODE}$ indicating which parents the node $\NODE$ should code over. Ideally, the required set would be equal to the parent set, but this may not be possible in all settings or applications. (See \sectionref{sec:applications} where we discuss various choices of the required set.) See \figureref{fig:network} for a diagram of a network using our notation.

\begin{figure}[h!]
\centering
\includegraphics[scale=0.46]{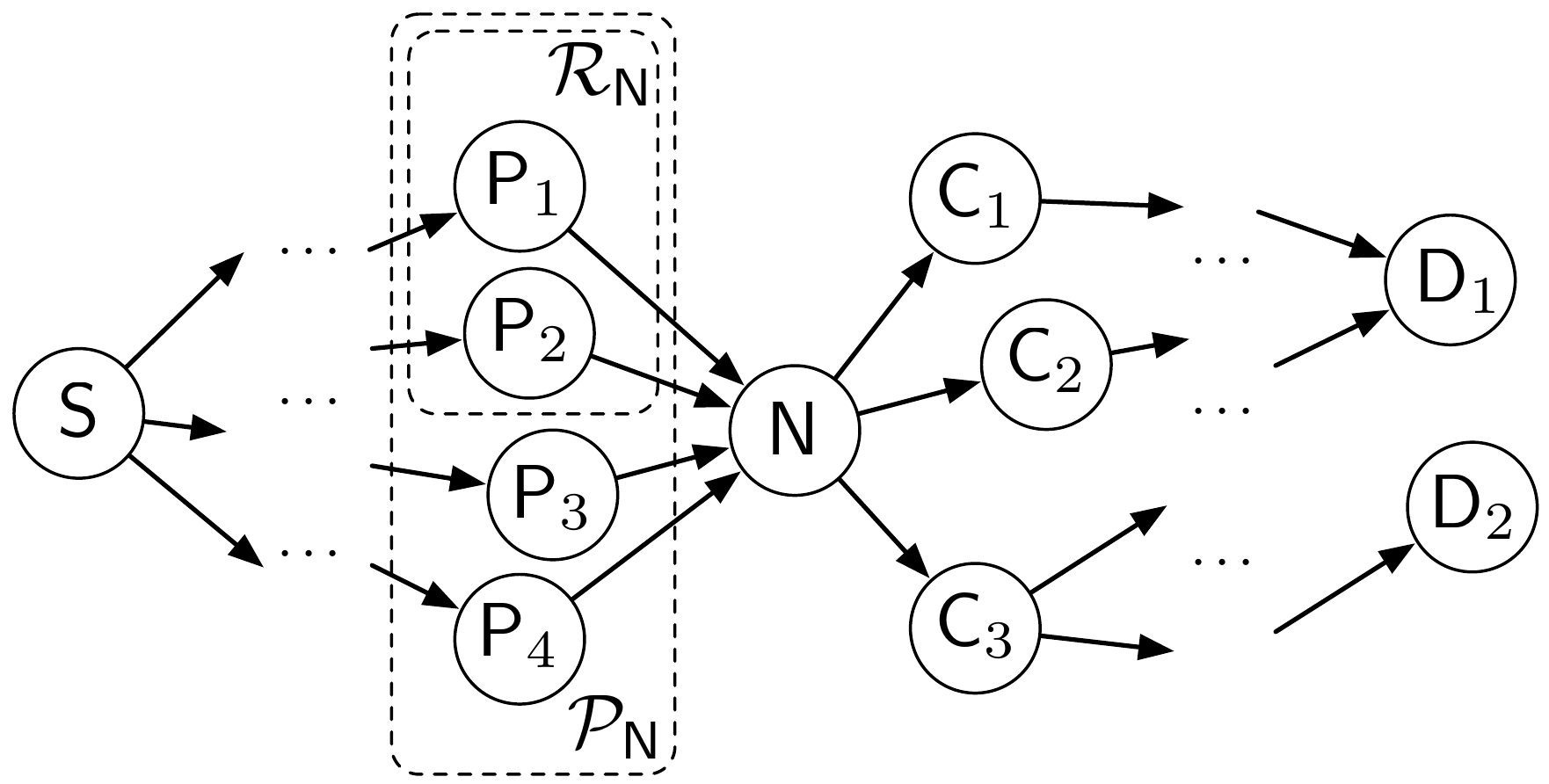}
\captionsetup{width=0.85\textwidth}
\caption{A source node $\SOURCE$ sends data to two destination nodes $\DESTINATION_{1}$ and $\DESTINATION_{2}$. A generic node $\NODE$ somewhere in the graph has parent nodes $\PARENT_{1}$, $\PARENT_{2}$, $\PARENT_{3}$, and $\PARENT_{4}$ (of which $\PARENT_{1}$ and $\PARENT_{2}$ form his required set $\REQUIREDSET_{\NODE}$) and has children nodes $\CHILD_{1}$, $\CHILD_{2}$, and $\CHILD_{3}$.}
\label{fig:network}
\end{figure}

Each node $\NODE$ has a public key $\pk_{\NODE}$ and a corresponding secret key $\sk_{\NODE}$. We assume that each node knows the public key $\pk_{\SOURCE}$ of the source $\SOURCE$; this is a reasonable assumption present in most previous work on pollution attacks \pollutionsignatures{}; for example, a node may be given this public key upon entering the system. 

In some settings (\sectionref{sec:applications}), we will need each node $\NODE$ to have a certificate $\cert_\NODE$ that his public key is valid and belongs to it; $\cert_\NODE$ consists of a signature from the source or some other trusted party: $\sig(\pk_{\NODE}, \text{``this is the public key of }N\text{''})$. A node need only obtain such a signature once per lifetime of the node and it can be performed, for example, when the node joins the network.

In order for a child $\CHILD_{j}$ to check that his parent coded correctly using the protocols that we present in \sectionref{sec:protocol},  $\CHILD_{j}$ needs to know what is the required set $\REQUIREDSET_{\NODE}$ of $\NODE$ and what are the public keys of the nodes in this set. Nodes do not need to know the required set (or the set of grandparents) for their parents a priori; in fact, dynamically adjusting the required set is important for dynamic networks. In \sectionref{sec:required-set}, we explain how nodes can acquire the required set for each of their parents depending on the application. We also explain for which applications our protocols are most fit and for which they are not fit.  For now, assume that each nodes knows precisely the nodes in the required set of each parent.

\subsection{Threat Model}
\label{sec:threat-model}

Nodes in the network may be \emph{Byzantine} (i.e., malicious or faulty): a node can pollute the data coming from the source by sending out a packet that is invalid or decrease the throughput by sending a packet that is not a result of coding over packets received from each parent in the required set. In \sectionref{sec:implementation-and-evaluation}, we discuss several Byzantine behaviors and how they affect the throughput of the network.

Even worse, Byzantine nodes can collude among each other. A node can collude with his parents, children or any other node in the network to pass the verification tests at his honest children.

We consider the adversarial model in which Byzantine nodes will use the best adversarial strategy to decrease the throughput at the sinks \textit{while still passing our verification tests}. As already discussed, we assume that there exist penalties in place that create enough incentives for \emph{not cheating detectably}; a discussion of what these penalties should be (e.g., a fine, an investigation, removal from the system, resource choking, reputation decrease, or making topology adjustments) is out of the scope of this paper and one should choose what best fits one's application.

\subsection{Solution Approach and Goals} 
\label{sec:our-solution-approach}

Similarly to prior work on pollution signatures, we also take a ``verification test'' solution approach. Our technical goal is to design a protocol that provably implements such a test for correctness:
\begin{center}
\begin{minipage}{0.85\textwidth}
\parhead{Verification test by node $\ch_j$ when receiving packet $P$ from node $\no$}
\emph{A procedure run by child $\ch_j$ upon receiving a packet $P$ from parent $\no$ to verify that node $\no$ generated $P$ by coding correctly (i.e., using pseudorandom coefficients over a packet from each parent in the required set $\REQUIREDSET_{\NODE}$ of $\NODE$).}
\end{minipage}
\end{center}
If a Byzantine node $\NODE$ passes the verification test performed by an honest child $\CHILD_{j}$, the Byzantine node must have coded correctly over the required data. Therefore, such a verification test would achieve the goal of this paper, because each honest node in the network has the ability to enforce correct random linear network coding at each of his parents.

Specifically, the verification test should satisfy the following properties:
\begin{enumerate}
  \item A Byzantine node that does not follow the random linear coding algorithm should be detected with overwhelming probability.
  \item The test must be efficient with respect to computation and bandwidth. 
  \item The verification test must be collusion resistant: an honest child should be able to check if his parent coded over all the \textit{honest} nodes in his required set, regardless of whether other children or grandparents are Byzantine or not.
  \item If the verification test fails, it is possible to prove it. In particular, this implies that a node can, not only detect, but also prove, when a parent cheats.
\end{enumerate}
We require that the computational overhead that each node incurs by running the verification test is reasonable and, moreover, we also require that the increase in packet size (due to the extra information sent to later nodes in order to enable them to run the verification test) does \emph{not depend on the payload} of the packet. (Recall that network coding is particularly useful when the packet payload is large and the overhead of the coefficients becomes negligible.) 

The protocols we propose (and which are presented in \sectionref{sec:protocol}) achieve the above four properties.

We remark that tackling collusion is challenging. For example, a node $\NODE$ could collude with a child $\CHILD_j$: $\NODE$ could send a packet to $\CHILD_{j}$ that is not the result of coding over all the nodes in the required set with pseudorandom coefficients, and $\CHILD_{j}$ would simply neglect running the verification test on $\NODE$. Still, we want to ensure that the \emph{other, honest} children of $\NODE$ can verify that they do receive correctly-coded packets. This means that each child node $\CHILD_{j}$ must be able to independently check $\NODE$ and not rely on any shared information that is required to stay secret. Similarly, ideally, if some parents collude with $\NODE$, $\NODE$'s children should still be able to check that $\NODE$ coded over all the required parents that did not collude with $\NODE$. This means that the parents cannot have some secret shared data in the protocol, all of this making the cryptographic protocol more challenging.

Finally, while the network model that we adopt is simple, we show in \sectionref{sec:applications} that it is \emph{expressive}: there we explain how to use this model for a variety of network settings and applications, either directly or with simple extensions.

\section{Protocol}
\label{sec:protocol}

We describe the protocols a node needs to run to perform the verification test on each of his parents and to assemble packets to send to his children. For clarity, we present the protocols in an incremental fashion, by successively adding more security properties. But first we will need to introduce some basic notation and cryptographic tools that we use.


\subsection{Notation}
\label{sec:notation}

A sequence (or tuple) of $n$ components $x_{1},\ldots,x_{n}$ is denoted by $(x_{1},\ldots,x_{n})$ or $(x_{i})_{i=1}^{n}$; for simplicity, sometimes we omit the starting and ending indices of the sequence, thus only writing $(x_i)_i$. The concatenation of two strings $a$ and $b$ is denoted by $a || b$.

We denote by $\INDEG{\NODE}$ the number of (parent) nodes in the required set $\REQUIREDSET_{\NODE}$ of node $\NODE$; by $\pk_{\NODE}$ and $\sk_{\NODE}$ the public and secret keys of node $\NODE$; and by $\sig_{N}(x)$ a signature of a message $x$ with respect to the key pair $(\pk_{\NODE},\sk_{\NODE})$ of $\NODE$, where the underlying signature scheme is assumed to satisfy the usual notion of unforgeability (i.e., existential unforgeability under chosen-message attack). For concreteness, we use the DSA algorithm \cite{FIPS-186-3}, whose signatures are only $320$ bits long.

Let $q$ be the prime number used in any of the pollution signature schemes in \pollutionsignatures. For example, in \cite{BonehFKW09}, $q$ is a $160$-bit prime. 

In network coding, as already mentioned, a packet has the form $E = \langle M,C \rangle$,  where $M$ is the \emph{payload} and $C$ the \emph{coding vector}. (In our protocols, we will augment the packet with additional tokens.) The payload $M$ is an $n$-tuple $(m_1,\ldots,m_n)$ of \emph{chunks}, where each chunk $m_{i}$ is an element of $\mathbb{Z}_{q}^{*}$, the multiplicative group of integers modulo the prime $q$. A coding vector $C$ is an $m$-tuple $(c_{1},\ldots,c_{m})$ of chunks, where each chunk is also an element of $\mathbb{Z}_{q}^{*}$.  Hence, $E$ consists of $n+m$ chunks $e_{1}, \ldots, e_{n+m}$, where $e_{i} = m_{i}$ for  $1 \leq i \leq n$ and $e_{i} = c_{i-n}$ for $n<i \leq n+m$. In particular, we can think of $M$, $C$, and $E$ as vectors in some product space of $\mathbb{Z}_{q}^{*}$.


\subsection{Cryptographic tools}
\label{sec:cryptographic-tools}

We now briefly review the cryptographic tools that we employ in our protocols:

\parhead{Pseudorandom functions}
Informally, a \emph{pseudorandom function family} is a family of polynomial-time computable functions $\{F_{s} \colon \{0,1\}^{|s|} \to \{0,1\}^{|s|}\}_{s \in \{0,1\}^{*}}$ with the property that, for a sufficiently large security parameter $k$ and a random $k$-bit seed $s$, $F_{s}$ ``looks like'' a random function to any efficient procedure.
See \cite{GoldreichGM86} for more details.

\parhead{Merkle hashes}
A \emph{Merkle hash}~\cite{Mer89} is a concise commitment to $n$ elements. Suppose that Alice has $n$ elements and she gives Bob a Merkle hash of them. Later, when Bob asks to see some elements from Alice, the Merkle hash allows Bob to check that indeed the elements Alice gives him are the same elements over which she had computed the Merkle hash. Loosely, to compute the Merkle hash of $n$ elements, Alice places the elements at the leaves of a full binary tree; she recursively computes each node higher in the tree as the hash of the concatenation of the two children. The resulting hash at the root is called the \emph{Merkle hash/commitment} of the $n$ elements. Given $n$ elements and their Merkle hash, Alice can reveal an element, say element $i$, to Bob by revealing the label of every node (and his sibling) along the path from the leaf node containing element $i$ to the root; Bob verifies the correctness of element $i$  by re-hashing the elements bottom-up and then verifying that the resulting hash is equal to the claimed Merkle hash. The advantage of the Merkle hash is that Bob only needs to ask $O(\log n)$ elements from Alice to check that a element out of $n$ has been correctly included in the Merkle hash. See ~\cite{Mer89} for more details.

\parhead{Pollution signatures}
A \emph{pollution signature scheme} (such as ~\pollutionsignatures{}) is a signature scheme consisting of the usual triplet of algorithms $(\gen,\sig,\ver)$ with a special homomorphic property that allows it to be used to detect pollution attacks in network coding.

Specifically, the source $\SOURCE$ runs the key generation algorithm $\gen$ to produce a secret key $\sk_{\SOURCE}$, together with a corresponding public key $\pk_{\SOURCE}$ that is published for everyone to use. The source $\SOURCE$ augments each outgoing packet $E$ with a special signature $\sigma_{\SOURCE}(E)$, generated by running the algorithm $\sig$ on input the secret key $\sk_{\SOURCE}$ and the packet $E$; we refer to this special signature as a \emph{validity signature} of the packet $E$ with respect to the public key $\pk_{\SOURCE}$.

When a node receives a (signed) packet $\langle E, \sigma_{\SOURCE}(E) \rangle$, he verifies the signature on the packet, by running the algorithm $\ver$ on input the public key $\pk_{\SOURCE}$, the packet $E$, and the signature $\sigma_{\SOURCE}(E)$.

Pollution signature schemes have the useful homomorphic property that, when given several packets together with their validity signatures, any node is able to compute a validity signature of any linear combination of those packets, \emph{without} communicating with the source $\SOURCE$. For example, if a node $\NODE$ receives two (signed) packets $\langle E_1, \sigma_{\SOURCE}(E_1) \rangle$ and $\langle E_2, \sigma_{\SOURCE}(E_2) \rangle$, then, for any two coefficients $\alpha$ and $\beta$, $\NODE$ can compute a validity signature of the packet $E = \alpha E_1 + \beta E_2$; in some schemes, this is done by computing  $\sigma_{\SOURCE}(\alpha E_{1} + \beta E_{2}) = \sigma_{\SOURCE}(E_1)^{\alpha} \cdot \sigma_{\SOURCE}(E_2)^{\beta}$, where each of these computations are performed in a certain field and the equality holds due to homorphism. See \pollutionsignatures{} for more details.


\subsection{A Generic Protocol}
\label{sec:highlevel}

In order to avoid repetition in the  presentation of our protocols, in this section we introduce the general structure that will be followed by each protocol version that we present; later, in any given protocol version, we will replace any unspecified quantities or procedures with concrete values or algorithms.

First we discuss the new packet structure: every packet $E$ transmitted by a generic node $\NODE$ is augmented with three cryptographic tokens; the first token has already been used in prior work, while the last two tokens are new to our protocols:
\begin{enumerate}
  \item A \emph{validity signature} $\sigma_{\SOURCE}(E)$, which is used to prevent pollution attacks. Any (secure) pollution signature scheme \pollutionsignatures{} may be used to produce this signature (as we rely only on the guarantees it provides and not on details of its implementation).
  \item A \emph{test token} $T_{\NODE}$, which is used by each child $\CHILD_{j}$ of $\NODE$ to run the verification test on $\NODE$, denoted $\vt$.
  \item A \emph{helper token} $H_{\NODE}$, which is used by each child $\CHILD_{j}$ of $\NODE$ to produce his own test token $T_{\CHILD_{j}}$, using a procedure called $\combine$. 
\end{enumerate}
Specifically, the protocol that a generic node $\NODE$ runs, after receiving packets from his parents, in order to produce a packet for each of his children, takes the general form of \algorithmref{algo:template}, where the procedures $\vt$, $\checkh$, and $\combine$, as well as the value of $H_{\NODE}$, will be specified later:

\begin{algorithm}[h!]
\caption{Protocol at a generic node $\NODE$}
\begin{algorithmic}[1]
\label{algo:template}

\STATE \label{step:packet}
       From each parent node $\PARENT_{i} \in \PARENTSET_{\NODE}$, node $\NODE$ receives a packet $\langle E_{\PARENT_{i}}, \sigma_{\SOURCE}(E_{\PARENT_{i}}), T_{\PARENT_{i}}, H_{\PARENT_{i}} \rangle$.

\STATE \label{step:checkvalid}
       For each parent node $\PARENT_{i} \in \PARENTSET_{\NODE}$, node $\NODE$ verifies that $\sigma_{\SOURCE}(E_{\PARENT_{i}})$ is a valid signature of $E_{\PARENT_{i}}$ using the public key $\pk_{\SOURCE}$ of the source $\SOURCE$, verifies that $\vt\left(E_{\PARENT_{i}}, \sigma_{\SOURCE}(E_{\PARENT_{i}}),T_{\PARENT_{i}}\right)$ accepts, and that $H_{\PARENT_i}$ is correct using $\checkh(E_{\PARENT_{i}}, \sigma_{\SOURCE}(E_{\PARENT_{i}}),H_{\PARENT_{i}})$ .

\STATE \label{step:computeE}
       Node $\NODE$ computes $E_{\NODE}$ by coding over all $E_{\PARENT_{i}}$ and $\sigma_{\SOURCE}(E_{\NODE})$, as described at the end of \sectionref{sec:cryptographic-tools}.

\STATE \label{step:computetokens}
       Node $\NODE$ computes $T_{\NODE} = \combine\big((E_{\PARENT_{i}}, \sigma_{\SOURCE}(E_{\PARENT_{i}}), H_{\PARENT_{i}})_{\PARENT_{i} \in \PARENTSET_{\NODE}}\big)$ and $H_{\NODE}$.

\STATE \label{step:send}
       Node $\NODE$ assembles the packet $\langle E_{\NODE}, \sigma_{\SOURCE}(E_{\NODE}), T_{\NODE}, H_{\NODE} \rangle$, and sends it to each child $\CHILD_{j}$. 
       
\end{algorithmic}
\end{algorithm}

In \stepref{step:checkvalid}, for each parent $\PARENT_i$ from which $\NODE$ receives a packet: $\NODE$ checks the validity signature of the packet to establish whether $\PARENT_i$ sent polluted data or not; then, $\NODE$ checks the test token $T_{\PARENT_i}$ by running the verification test to establish that $\PARENT_i$ coded correctly; next, $\NODE$ needs to make sure $\PARENT_i$ sent a correct helper token (without which $\NODE$ could not compute a good test token $T_{\NODE}$ himself and would fail the verification test at his children).

If any of the checks above fail, $\NODE$ will report them and act in some way that is application-specific. As we will see in \sectionref{sec:prove-when-cheat}, $\NODE$ can accompany his complaint with a proof that his parent cheated.

In our protocol, each node verifies his parents (if he is not the source) and is being verified by his children (if he is not a sink/destination). Thus, $\NODE$ verifiers $\PARENT_i$, and $\CHILD_j$ verifies $\NODE$.


\subsection{How to Force Byzantine Nodes to Code Over All Required Packets}
\label{sec:force-valid}

As a first step, we design a verification test that enables any child $\CHILD_{j}$ of a node $\NODE$ to check that $\NODE$ did indeed code over all of the parent nodes in his required set $\REQUIREDSET_{\NODE}$, i.e., that the packet sent by $\NODE$ to $\CHILD_{j}$ is a linear combination of packets from parents in the required set with coefficients all of which are not equal to zero.

\parhead{A na\"ive solution}
The node $\NODE$ can simply forward to each of his children all the packets received from parents in the required set. Of course, $\NODE$'s parents make sure to sign (using their own secret keys) the packets they send to $\NODE$, so that $\CHILD_{j}$ can be sure that the packets forwarded by $\NODE$ are indeed from $\NODE$'s parents. In other words, $\NODE$ forwards to each child $\CHILD_j$ the following data: $E_{\PARENT_{i}}$, $\sigma_{\SOURCE}(E_{\PARENT_{i}})$, and $\sig_{\PARENT_{i}}(E_{\PARENT_{i}} || \sigma_{\SOURCE}(E_{\PARENT_{i}}))$,  the coding coefficients used for the packets from each parent, and the newly coded payload  $E_{\NODE}$ with the new integrity signature $\sigma_{\SOURCE}(E_{\NODE})$. Each child $\CHILD_{j}$ can then establish whether $\NODE$ coded correctly, because he now has access to all the information $\NODE$ received from his parents and can thus check that $\NODE$ did not use any zero coefficients. 

Clearly, this solution is bandwidth inefficient: the payload of the packet can be very large and $\NODE$ will send $\INDEG{\NODE} + 1$ such payloads to his children, reducing throughput $\INDEG{\NODE}+1$ times. 

\parhead{Payload-Independent Protocol (\one)}
We now improve on the na\"ive solution, by avoiding to include the packet payload in the test token sent for verification, thus saving considerable bandwidth and throughput.
 
Each parent $\PARENT_{i}$ sends a helper token consisting of a parent signature on the validity signature:
\begin{equation*}
H_{\PARENT_{i}} := \sig_{\PARENT_{i}}\big( \sigma_{\SOURCE}(E_{\PARENT_{i}}) \,||\, \text{``from $\PARENT_{i}$ to $\NODE$''} \big) \enspace.
\end{equation*}
The text in $H_{\PARENT_{i}}$ prevents a colluding $\NODE$ from giving this helper token to some other node $\NODE'$, which could otherwise falsely claim that he received the data from $\PARENT_{i}$.

The test token $T_{\NODE}$ of node $\NODE$ is computed by a simple concatenation; specifically, $\combine$ computes the following test token:
\begin{equation*}
T_{\NODE} = \Big(\big(\alpha_i, \sigma_{\SOURCE}(E_{\PARENT_{i}}), H_{\PARENT_{i}}\big)\Big)_{\PARENT_{i} \in \PARENTSET_{\NODE}} \enspace,
\end{equation*}
where $\alpha_{i}$ is the coding coefficient that $\NODE$ used for the packet from $\PARENT_{i}$.

The verification test for this version of the protocol is given in \algorithmref{algo:checknoMerkle}. 

\begin{algorithm}[h!]
\caption{$\vt$ of node $\NODE$, by child $\CHILD_{j}$}
\label{algo:checknoMerkle}
\begin{algorithmic}[1]

\FOR{each parent $\PARENT_{i} \in \PARENTSET_{\NODE}$}
  
  \STATE $\CHILD_{j}$ checks that $T_{\NODE}$ contains an entry $\big(c_{i}, \sigma_{\SOURCE}(E_{\PARENT_{i}}), H_{\PARENT_{i}}\big)$ for the current parent $\PARENT_{i}$.
 
  \STATE $\CHILD_{j}$ checks that $H_{\PARENT_{i}}$ verifies with $\pk_{\PARENT_{i}}$ as a signature of $\sigma_{\SOURCE}(E_{\PARENT_{i}})$. \label{step:verifysig}
   
  \STATE $\CHILD_{j}$ checks that $\alpha_{i} \neq 0$. \label{step:nonzero}

\ENDFOR
  
  \STATE \label{step:combine}
         $\CHILD_{j}$ verifies that combining all validity signatures $\sigma_{\SOURCE}(E_{\PARENT_{i}})$ and coefficients $\alpha_{i}$ results in $\sigma_{\SOURCE}(E_{\NODE})$, according to the homomorphic property discussed in \sectionref{sec:cryptographic-tools}.

  \end{algorithmic}  
\end{algorithm} 

Step 2 verifies that $\NODE$ provided test data for parent $\PARENT_{i}$. Step 3 checks that the data is authentic. Step 4 establishes that the coefficient used in coding over this parent is nonzero. Step 6 checks that the coded data from $\NODE$ indeed corresponds to coded data over the information from the parents with the claimed coefficients $\alpha_{i}$.  
 
We now give some intuition for why \algorithmref{algo:checknoMerkle} is a good verification test, leaving a formal proof to \sectionref{sec:proofs}. If $\NODE$ does not code over the packet of some parent, say from parent $\PARENT_{1}$, in order for $\NODE$ to produce a validity signature $\sigma_{\SOURCE}(E_{\NODE})$ for $E_{\NODE}$ that verifies successfully under the public key $\pk_{\SOURCE}$, $\NODE$ needs to combine only those validity signatures from parents he coded over and not include $\sigma_{\SOURCE}(E_{\PARENT_{1}})$ in the computation; at \stepref{step:combine}, however, $\CHILD_{j}$ uses the validity signature from all required parents to check $\sigma_{\SOURCE}(E_\NODE)$ (with a coefficient $\alpha_i$ that was checked to be nonzero in \stepref{step:nonzero}) and the check would fail.
 
  $\checkh$ at $\CHILD_j$ consists of checking that $H_{\NODE}$ is indeed a signature on $E_{\NODE}$ and is not a signature on zero.
 
The length of the test token $T_{\NODE}$ is now
\begin{equation*}
|T_{\NODE} | = \INDEG{\NODE} \cdot \big( |\sigma_{\SOURCE}| + |\sig| \big) \enspace,
\end{equation*}
where $|\sigma_{\SOURCE}|$ denotes the size of the pollution signature and $|\sig|$ the size of the signature scheme introduced in \sectionref{sec:notation}. Indeed, the length of $T_{\NODE}$ does not depend on the payload any more. Also, recall that the lengths of the signatures are constant. Note, though, that $|T_{\NODE}|$ is linear in the number of parents; this may not be a problem, but in applications where the payload is not that large or where there can be many parents, it would be desirable to have a smaller token. Moreover, verifying $\INDEG{\NODE}$ digital signatures in the verification test (\stepref{step:verifysig} above) will become expensive if the number of parents is not small.

\parhead{Logarithmic Payload-Independent Protocol (\two)}
We provide a second protocol in which the length of the helper token $T_{\NODE}$ is significantly shorter:
\begin{equation*}
|T_{\NODE}| = |h| + |\sigma_{\SOURCE}| + |\sig| + 2 \cdot |\sigma_{\SOURCE}| \cdot \log(\INDEG{\NODE}) \enspace,
\end{equation*}
where $|h|$ is the size of a hash (e.g. $160$ bits for SHA-1), thus replacing $\INDEG{\NODE}$ with the much smaller value $\log(\INDEG{\NODE})$ (where the logarithm is base $2$). The second protocol that we present, however, is probabilistic in its guarantees: rather than enabling a child $\CHILD_{j}$ to test if a parent $\NODE$ cheated (with overwhelming confidence), we enable $\CHILD_{j}$ to detect misbehavior of $\NODE$ with a certain (adjustable) probability.

Specifically, after receiving the packet from $\NODE$, node $\CHILD_{j}$ picks a required parent of $\NODE$ at random and challenges $\NODE$ to prove that he coded correctly over that parent. Of course, $\NODE$ does not know ahead of time on what packets he will be challenged.  As shown in \sectionref{sec:implementation-and-evaluation}, such probabilistic approach is still quite effective because the chance that a Byzantine node $\NODE$ is detected cheating grows exponentially in the number of times he attempts to cheat. In \sectionref{sec:implementation-and-evaluation}, we provide recommendations for when we believe it is more appropriate to use \one{} or \two.

The basic idea of \two{} is that $\NODE$ will send to $\CHILD_{j}$ a test token $T_{\NODE}$ that is the root of a Merkle hash tree constructed over the data of the test token used in \one{}; namely, a Merkle hash tree where the elements at the leaves are the tuples $\big(\alpha_i, \sigma_{\SOURCE}(E_{\PARENT_{i}}),  \sig_{\PARENT_{i}}\big( \sigma_{\SOURCE}(E_{\PARENT_{i}}) \,||\, \text{``from $\PARENT_{i}$ to $\NODE$''} \big)\big)$ ranging over the parents $\PARENT_{i}$. Each $\CHILD_{j}$ will then challenge $\NODE$ by asking to see a certain path in the Merkle hash tree corresponding to a parent of $\NODE$. In this way, $\CHILD_j$ can check if $\no$ coded over that parent (i.e., $\NODE$ used a non-zero coefficient). Of course, $\NODE$ cannot provide arbitrary data to $\CHILD_{j}$ when replying the challenge of $\CHILD_{j}$, as guaranteed by the security properties of a Merkle hash. Therefore, if $\NODE$ did not code over a parent, $\CHILD_{j}$ will discover this with a known probability.

\begin{figure}[t!]
\centering
\includegraphics[scale=0.42]{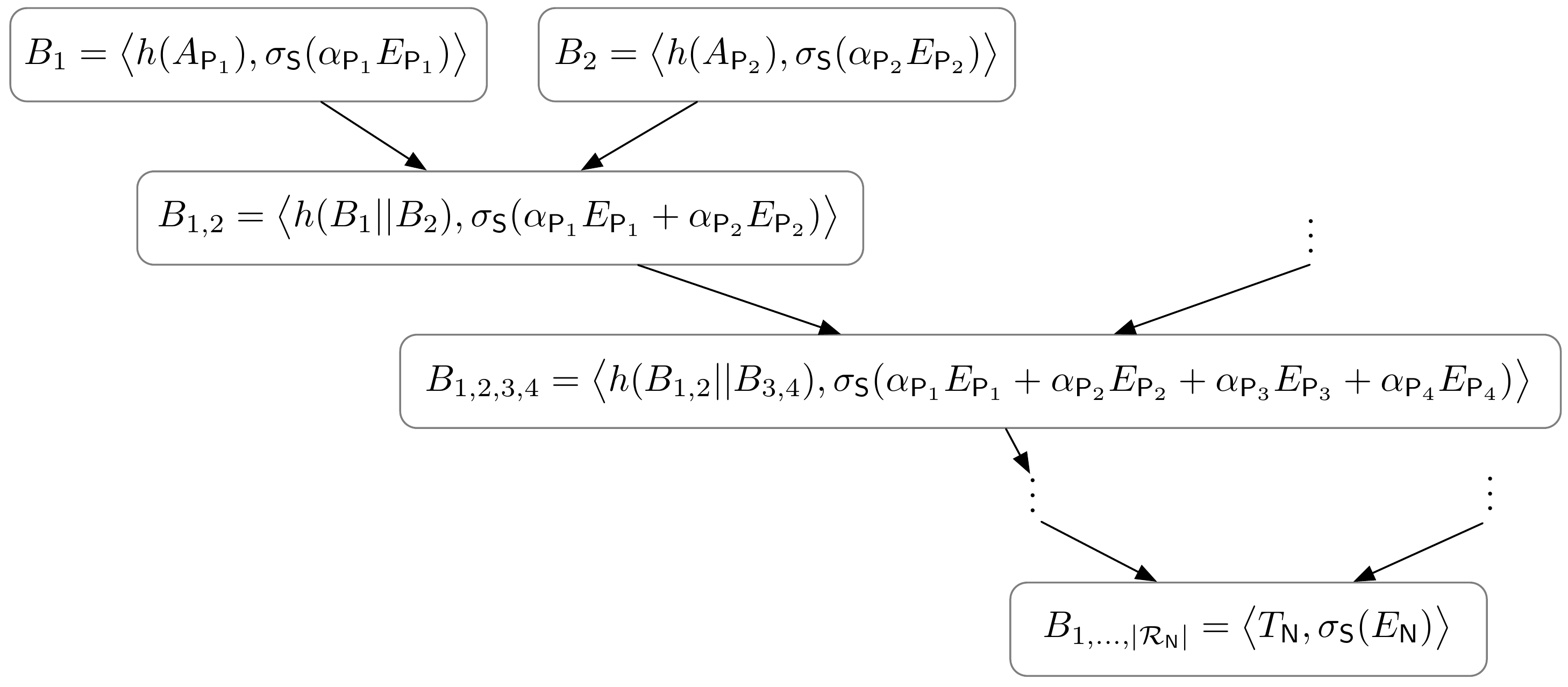}
\captionsetup{width=0.85\textwidth}
\caption{Diagram representing the computation of the helper token $T_{\NODE}$: for each parent $\PARENT_{i}$, define $A_{\PARENT_i} = \sigma_\mathsf{S}(E_{\PARENT_i}) \vert\vert H_{\PARENT_i} \vert\vert \alpha_{\PARENT_i}$; then recursively apply the hash function $h$ as indicated, and recursively compute pollution signatures with the indicated coefficients.}
\label{fig:merkle}
\end{figure}

Let $h$ be a hashing scheme. \figureref{fig:merkle} illustrates the Merkle tree that $\NODE$ has to compute and provides notation for our discussion. We slightly modify the traditional Merkle hash, by adding data at internal nodes and changing the recursion.  Each leaf node $A_{\PARENT_{i}}$ in the Merkle tree consists of a ``summary'' of the data from a required parent $\PARENT_{i}$ of $\NODE$. Each internal node consists of the validity signature obtained by coding over all the packets at the leaves of the subtree rooted by the internal node and a hash of the two children. The root node will thus contain the test token $T_{\NODE}$ as the root hash and the validity signature over $E_{\NODE}$, namely $\sigma_{\SOURCE}(E_{\NODE})$. 
Thus, $\combine$ consists of computing the Merkle hash to obtain the Merkle hash root, so that $T_{\NODE} = \text{``Merkle hash root''}$. $H_{\PARENT_{i}}$ is the same as in \one, that is, $H_{\PARENT_{i}} = \sig_{\PARENT_{i}} \big( \sigma_{\SOURCE}(E_{\PARENT_{i}}) \,||\, \text{``from $\PARENT_{i}$ to $\NODE$''} \big)$.

The verification test $\vt$ is ran in a different way than in \one. Each node $\CHILD_{j}$ receives the packet from $\NODE$, checks the validity signature, and he can proceed to code and forward the packet. It can then challenge node $\NODE$ to check if $\NODE$ indeed coded over all the packets.  During a challenge, only $\log \INDEG{\NODE}$ source signatures and hashes will be retrieved, due to the Merkle tree property. Moreover, only one digital signature will be verified, the one corresponding to the parent from the challenge. For the Merkle recursion, only hash verifications and  homomorphic signature operations (which typically consist of multiplying 1024-bit numbers) will be performed, so the overall cost is dominated by a single digital signature verification. 

The number of challenges is selected based on the desired probability of detection. With $t$ challenges, there is a probability of $t/\INDEG{\NODE}$ of detecting that $\no$ did not code over a parent. After $r$ transmissions in which $\no$ cheats, the probability of detecting $\no$ is at least $1 - (1-t/d)^r$ (this is achieved when $\no$ cheats minimally -- by not coding over one parent), which increases exponentially in $r$. Coupled with penalties, such a probabilistic approach offers incentives against cheating.

Node $\NODE$ needs to remember the values that constituted the Merkle tree until the children finished challenging it. 
One challenge checks that the node coded correctly over a parent; multiple challenges can be sent at once and processed together. 

\begin{algorithm}
\caption{Challenge on node $\NODE$ by node $\CHILD_{j}$}
\label{algo:challenge}
\begin{algorithmic}[1]

  \STATE  $\CHILD_{j}$ picks a parent $\PARENT_{i}$ of $\NODE$ at random and informs $\NODE$ of the choice.

  \STATE $\NODE$ must present $A_{\PARENT_{i}}$ (defined in \figureref{fig:merkle}) and all values of the nodes in the Merkle tree that are siblings to nodes on the path from $A_{\PARENT_{i}}$ to the root and their siblings. 

  \STATE \label{step:verif} $\CHILD_{j}$ runs $\vt$:  \begin{enumerate} [nolistsep,label*={\small \roman*}]
  \item verifies that $H_{\PARENT_{i}}$ is a correct signature over $\sigma_s(E_{\pa_i})$ using $\pk_{\PARENT_{i}}$ and that $\alpha_i \neq 0$ \label{step:coeff}
  \item verifies that the validity signature in $A_{\PARENT_{i}}$ combined with $\alpha_{i}$ is the same as the validity signature in $B_{i}$, 
  \item the validity signature at each internal node is a multiplication of the validity signatures at the children of the node, 
  \item recomputes the Merkle hash based on the hashes provided by $\no$ and checks equality to $T_{\NODE}$,
  \item checks that the validity signature, provided when $\NODE$ initially transmitted, verifies the validity signature at the top of the Merkle tree,
  \item checks $H_{\NODE}$ to be a signature using $\pk_{\NODE}$ on $\sigma_{\SOURCE}(E_{\NODE})$.
\end{enumerate}
\end{algorithmic}
\end{algorithm}


As for $\checkh$, $\CHILD_{j}$ still needs to check that $H_{\NODE}$ is indeed a signature on $E_{\NODE}$ and is not a signature on zero to prevent $\NODE$ from causing $\CHILD_{j}$ to fail the verification test at $\CHILD_{j}$'s children. Proofs of security for this protocol are included in \sectionref{sec:proofs}.

\parhead{Collusion}
Both \one{} and \two{} are collusion resistant: even if a child colludes with $\NODE$, the other children check $\NODE$ independently. Moreover, if $\NODE$ colludes with a parent $\PARENT_{1}$, $\NODE$ still needs to code correctly over the rest of the parents that he did not collude with because he cannot forge these parents signatures if $\NODE$ has at least one honest child verifying it.

\subsection{How to Force Byzantine Nodes to Code Pseudorandomly}
\label{sec:force-diverse}

As a second step, we design a verification test that enables any child $\CHILD_{j}$ of a node $\NODE$ to check, not only whether a packet received from $\NODE$ is valid and derived using non-zero coefficients over each parent in his required set (as was guaranteed by the solution presented in \sectionref{sec:force-valid}), but also whether the node $\NODE$ coded using (pseudo)random coefficients in $\mathbb{Z}_{q}^{*}$.

The basic idea is to require node $\NODE$ to generate the pseudorandom coefficients from a seed that is also known to each child $\CHILD_{j}$, so that each $\CHILD_{j}$ will be able to generate these same coefficients and use them as part of his verification test on $\NODE$.

We assume that each client knows a random seed $s$ that is public; a trusted party drew the seed at random when the system started. For example, a client can learn about the seed $s$ when he joins the system. In a wireless setting with no membership, a node can either have $s$ already hardcoded, or he can obtain it from his neighbors ($s$ can be accompanied by a signature from a trusted party to make sure that malicious neighbors cannot lie about its value). The seed can remain the same for the lifetime of the system.

Using the seed $s$, the coefficients can then be generated using a pseudorandom function $F_{s}$ (defined in  \sectionref{sec:cryptographic-tools}). 
For each parent $\PARENT_{i}$ in the parent set $\PARENTSET_{\NODE}$, the node $\NODE$ computes $\alpha^*_i = F_{s}(\PARENT_{i} \vert\vert \NODE)$ (of course, mapped to the field of the coefficients) and uses $\alpha^*_i$ as the coding coefficient for the packet from $\PARENT_{i}$. 

Observe that, contrary to what the definition of the pseudorandomness property \cite{GoldreichGM86} prescribes, the seed $s$ is \emph{not} kept private, but is instead made \emph{public}. Of course, in such a case, one cannot expect that the input-output relation induced by $F_{s}$ is unpredictable; indeed, it is deterministic, because now $F_{s}$ may be computed by anyone (and is not an ``oracle'' anymore). Nonetheless, since in our setting the inputs to $F_{s}$ are not under the control of Byzantine nodes, and are predetermined, it is easy to show that the outputs of $F_{s}$ on these inputs will still retain the statistical properties that we are interested in, allowing for the network throughput to still be maximal using these ``pseudorandom'' coefficients.

If one wishes to enable $\NODE$ to use a different set of coding coefficients for each child $\CHILD_{j}$, the computation of the coding coefficients can be changed to $\alpha^*_{i,j} = F_{s}(\PARENT_{i} \vert\vert \NODE \vert\vert \CHILD_{j} )$; thus $\NODE$ must use $\alpha^*_{i,j}$ to code over the data from $\PARENT_{i}$ when preparing a packet for child $\CHILD_{j}$. Intuitively, using different coefficients increases throughput in some topologies because of more diversity; this can be helpful in P2P networks, for example, but not so much in a wireless setting where transmitting different data to children will not take advantage of the shared medium on which multiple children can listen.

The verification tests in previous sections can now be easily modified to have each child $\CHILD_j$ check that $\no$ coded over each parent in the required set with this exact coding coefficients $\alpha^*_i$ (or $\alpha^*_{i,j}$): in \stepref{step:nonzero} of \algorithmref{algo:checknoMerkle} and in \stepref{step:verif}\ref{step:coeff} of \algorithmref{algo:challenge}, node $\CHILD_{j}$ must check that $\alpha_i$ equals $\alpha^*_i$ (or $\alpha^*_{i,j}$).
With this check in place, Byzantine nodes are forced to code with pseudorandom coefficients.  \sectionref{sec:proofs} shows that Byzantine nodes cannot code with different coefficients and pass the verification test.

\subsection{How to Prevent Replay Attacks of Old Data}
\label{sec:prevent-replay}

One problem is that a Byzantine client may code correctly for one transmission, but may attempt to cheat on the next transmission by sending the old data he sent for the first transmission. In some cases, such a strategy reduces throughput, but in others, it even pollutes packets downstream in the network. Nevertheless, the Byzantine client will pass any pollution test because the source uses the same keys for signing in both transmissions; the node will also pass our diversity tests above because he coded correctly over his parents in the first transmission.

Therefore, we need to prevent such replay attacks. In fact, the problem of replay attacks belongs to the use of pollution schemes and is not introduced by our diversity enforcement scheme. Any solution for that setting will suffice in our setting as well because of the way we build ``on top'' of validity signatures. Therefore, any overhead introduced by such a scheme already is introduced by the use of pollution schemes and does not come with diversity enforcement. 

We propose one such replay solution. 
The idea is to have the source change the validity signature key with every transmission so that any attempt by a Byzantine client to use old data would be detected when checking the validity signature. Let $(\sk_{S, k}, \pk_{S, k})$ denote the public key used by the source in the $k$-th transmission.  The source has one master signing key pair of which the public verification key is known to all users as before. To inform nodes of the public key used during a transmission, the source will send with every packet this public key accompanied by a signature of this public key using the master signing key. The source signs the public key to prevent malicious clients from forging public keys of their own and claiming they belong to the source.




For our diversity scheme, we make use of the public key corresponding to each transmission to add diversity in the coding coefficients across transmissions. Each node should now code with  $\alpha^*_i = F_{s}(\PARENT_{i} \vert\vert \NODE \vert\vert \pk_{S,k})$ and their children will check the inclusion of $\pk_{S,k}$ in the coding coefficients along with the other tests they perform; without $\pk_{S,k}$, the coding coefficients will be the same across different transmissions.



\subsection{How to Enable Nodes to Prove Misbehavior}
\label{sec:prove-when-cheat}

We discuss how any child $\CHILD_{j}$ of a node $\NODE$ can \emph{prove} $\NODE$'s misbehavior to a third party, when the verification test for $\no$ fails. Recall that the ability to convince a third party (such as the source, a membership service, or other authoritative agents in the system) that $\NODE$ did indeed misbehave is important to allow for punitive measures to be enacted. Furthermore, the ability to prove misbehavior \emph{reinforces the deterrent effect of verification tests}.

We use signatures in a natural way to provide such proofs: \stepref{step:send} of \algorithmref{algo:template} is modified so that a node $\NODE$ attaches an additional ``attest'' token to the packet he sends to his children; the attest token consists of a signature of the whole packet under his own secret key $\sk_{\NODE}$. Each child $\CHILD_{j}$ of $\NODE$ will then verify this signature (and ignore any data from $\NODE$ that does not carry a valid ``attest'' signature). 

If a child $\CHILD_{j}$ establishes that his parent $\NODE$ did not code correctly based on the verification tests in \algorithmref{algo:checknoMerkle} or \algorithmref{algo:challenge}, he can provide the packet from $\NODE$ together with his attest token as proof to a third party. Any other party knowing the required set $\REQUIREDSET_{\NODE}$ of node $\NODE$ can run the $\vt$ procedures to establish if $\NODE$ cheated. 
Of course, by the unforgeability property of the signature scheme, children of $\NODE$ cannot falsely accuse $\NODE$ of misbehavior.

\subsection{Proofs of Security}
\label{sec:proofs}

\begin{theorem}[Security of \one]
\label{thm:security-one}
\emph{In protocol \one{}, if a generic node $\NODE$ passes all checks at an honest child $\CHILD_{j}$, it means that $\NODE$ coded over the value from $\pa_1$, $E_{\pa_1}$, with precisely coefficient $c_1$ (as described in \sectionref{sec:force-diverse}), where $\pa_1$ is any generic parent from $\NODE$'s required set. }
\end{theorem}
\begin{proof}
\algorithmref{algo:checknoMerkle} gives $\vt$ for the \one{} protocol. If $\NODE$ passes the checks in Step 2, it means that $\NODE$ provided the triple $\big(c_1, \sigma(E_{\PARENT_1}), H_{\PARENT_1}\big)$ in $T_{\NODE}$; if $\NODE$ passes the checks in tep 3 and Step 4, it means that $\pa_1$ indeed provided $\sigma(E_{\pa_1})$ and $c_1 \neq 0$; if $\NODE$ passes the check in Step 6, it means that $\NODE$ computed $\sigma(E_{\NODE})$ by including $\sigma(E_{\pa_1})$ with coefficient $c_1$ in the homomorphic computation (described in \sectionref{sec:cryptographic-tools}).

In Step 2 of \algorithmref{algo:template} when run by $\CHILD_{j}$, the node $\CHILD_{j}$ checks that $\sigma(E_{\NODE})$ verifies as a signature of $\NODE$. By the theorem's hypothesis, the pollution signature verifies, so that, by the security of the pollution scheme (detailed in~\cite{BonehFKW09}), it must be the case that $\NODE$ included $c_1 \cdot E_{\pa_1}$ when computing $E_{\NODE}$.  
\end{proof}

\begin{theorem}[Security of \two]
\label{thm:security-two}
\emph{In protocol \two, if a generic node $\NODE$ did not code over any given parent, say $\pa_1$, from his required set with coefficient $c_1$ (as described in \sectionref{sec:force-diverse}), and an honest child $\CHILD_{j}$ challenges $\NODE$ on $t$ random parents, the probability that $\NODE$ is detected (some check fails) is at least $t/\INDEG{\NODE}$.}
\end{theorem}

\begin{proof}
The strategy of the proof is to present some exhaustive cases in which $\NODE$ could not have coded over a parent, and show that in each such case the probability of detection is $\ge t/\INDEG{\NODE}$. Consider the tree $\mathsf{Tree}_{\NODE}$ of values that $\NODE$ used when he computed the Merkle hash that he gave to $\CHILD_{j}$. Because of the Merkle hash guarantees, $\NODE$ cannot come up with any other tree (that is not a subtree of $\mathsf{Tree}_{\NODE}$), that has the same Merkle root hash. If any leaf $i$ in this tree (if a leaf exists) does not satisfy check (i) in Step 3 of \algorithmref{algo:challenge}, it will be caught if $\CHILD_{j}$ challenges $\NODE$ on parent $\PARENT_{i}$, which happens with probability $t/ \INDEG{\NODE}$. Similarly, if any internal nodes $B_i$ do not satisfy check (ii), $\CHILD_{j}$ will detect this with probability at least $t/\INDEG{\NODE}$. Therefore, we can assume that the first level of internal nodes in the $\mathsf{Tree}_{\NODE}$ consists of the expected hashes and $\sigma(c_i E_{\pa_i})$ where $c_{i}$ is the desired coefficient and $\sigma(E_{\pa_i})$ is indeed the validity signature from parent $\PARENT_{i}$. If any internal node in $\mathsf{Tree}_{\NODE}$ does not satisfy check (iii), this will be detected whenever $\NODE$ is challenged on a value $i$ that involves a path through the Merkle tree passing through the broken internal node; this happens with probability at least $t/\INDEG{\NODE}$. Therefore, assuming all internal nodes pass check (iii), it means that the validity signature at the top of the tree must be $\sigma(\sum_i c_i E_{\pa_i})$. If the validity signature at the top of $\mathsf{Tree}_{\NODE}$ does not match the one initially provided by $\NODE$ (i.e., $\sigma(E_{\NODE})$), check (v) will fail with probability 1. Assuming, this check succeeds it must be the case that the validity signature initially provided by $\NODE$ is a proper validity signature after coding with $c_{i}$ over all $\PARENT_{i}$. Since the validity signature matched $E_{\NODE}$ (check (2) of \algorithmref{algo:template} when run at child $\CHILD_{j}$), it means that $\NODE$ coded over all parents with the right coefficients, by the guarantees of the validity signature. Therefore, there are no more cases of possible cheating from $\NODE$ to consider and since all previous types of cheating were caught with chance $\ge t/\INDEG{\NODE}$, the proof is complete.
 \end{proof}



\section{Applications and Extensions}
\label{sec:applications}

In this section, we describe applications and extensions of our protocol. 

\subsection{Types of Required Sets}

In our protocols so far, we considered that a child of node $\NODE$ performs the verification test on a specific set of required parents for $\NODE$. However, one can use different types of verification tests, some being more useful for certain settings, as we will see. All these verifications, in fact, just map to verifying a specific required set as before.  

A child $\CHILD_j$ can perform any of the following checks for node $\NODE$:

\newcounter{reqcount}

\begin{list}{(\arabic{reqcount})}  {\usecounter{reqcount} \setlength{\leftmargin}{1em}}

\item \label{item:all}
        \emph{$\NODE$ coded over all his parents or over a specific set of parents.} 
        
\item \label{item:threshold}
        \emph{Threshold enforcement: $\NODE$ coded over at least $d$ parents.} This check can be enforced by having $\NODE$ send an indication of which parents he coded over with their public keys and certificates (defined in \sectionref{sec:threat-model}): $\CHILD_{j}$ checks that these are at least $d$ in number, checks the certificate of each public key to make sure $\NODE$ did not falsify these keys, and that $\NODE$ indeed coded over them.

\item \label{item:subset}
        \emph{$\NODE$ coded over at least some subset of parents.} This is a combination of \itemref{item:all} and \itemref{item:threshold}. $\CHILD_j$ checks that $\NODE$ coded over the subset of parents as in \itemref{item:all} and over some valid parents as in \itemref{item:threshold}. 
        

\item \label{item:app-policy}
        \emph{$\NODE$ coded over a set of parents with some application-level property.} For example, $\NODE$ must code over at least two parents noted by some application as high priority and at least five parents in total. The priority of each node $\PARENT_i$ can be included in the certificate $\cert_{\PARENT_i}$. $\NODE$ again indicates the nodes he coded over to $\CHILD_{j}$ along with their public keys and certificates, and $\CHILD_j$ checks that at least two certificates contains high priority and there are at least five in total.  Other general application semantics can be supported by this verification case.

\end{list}


\subsection{Applications and Required Sets}
\label{sec:required-set}

In this section, we describe the various settings to which our protocols are applicable, and how the nodes would learn of the required set of their parents. 

Our model applies to settings in which a node can learn the required set of his parents, such as:

\newcounter{appcount}

\begin{list}{\arabic{appcount})}  {\usecounter{appcount} \setlength{\leftmargin}{1em}}

\item \label{c:memb} \textit{Systems with a membership service:} the membership service can inform a node of his grandparents when the node joins and when changes occur. Some peer-to-peer and content distribution systems fall in this category.

\item \label{c:link} \textit{Systems having a reliable yet potentially low capacity channel} besides  the channel where the coding occurs (which may be less reliable, but has higher capacity): the reliable channel can be used to communicate topology changes between nodes. Some examples of applications are decentralized peer-to-peer applications and content distribution, as well as some wireless networks.

\item \label{c:static} \textit{Static topologies:} these topologies do not change or change rarely. The topology is mostly known to the nodes (e.g., nodes can discover it when joining), so a node will know his grandparents. Wired as well as some wireless network applications fall in this category.
For wired networks, since the topology is more static and delays tend to be lower, more aggressive verification tests can be implemented (e.g. the required set is most of the parents or all of the parents, depending on the particular system).

\item \label{c:dynamic} \textit{Moderately dynamic wireless topologies:} the set of grandparents for a node may change many times, after each change, it remains the same  for enough time allowing the node to discover the new grandparents.  

Let us discuss how a child can learn about his changing grandparents in dynamic topologies. First of all, for such topologies, we recommend nodes use the threshold enforcement scheme (described in (\itemref{item:threshold} above) because the set of parents of a node changes dynamically. The threshold should be adjusted based on some minimum number of links a node is expected to have in order to code diversely. 

Consider that parents of node $\NODE$ have changed and child $\CHILD_j$ wants to learn about this. We use the same links used by packet flow to inform $\CHILD_j$ of his grandparents. Each new parent $\PARENT_i$ sends $\NODE$: his public key and the corresponding certificate $\cert_{\PARENT_i}$. $\NODE$ sends this information to $\CHILD_i$. Let's discuss the case when $\NODE$ is malicious and may try to inform $\CHILD_i$ of incorrect parent list. Note that $\NODE$ cannot lie that $\PARENT_i$ is a parent when he is not because, if $\NODE$ does not have a link to $\PARENT_i$, during transmission time, nodes $\CHILD_i$ will verify that $\NODE$ coded over the data from $\PARENT_i$ which $\NODE$ could not have done because he did not receive this data. Moreover, $\NODE$ cannot create some public keys of his own and claim that some parents with those public keys exist, because each node key has a certification as discussed. On the other hand, $\NODE$ may try to simply not report any of his parents so that he does not have to forward or code over any data. However, each child $\CHILD_i$ will expect $\NODE$ to report at least a threshold of parents; if $\NODE$ does not do so, $\CHILD_i$ can be suspicious and denounce $\NODE$ of potentially being malicious, as discussed in \sectionref{sec:threat-model}. Therefore, \textit{$\NODE$ can choose which $d$ parents to code over from the set of parents physically linked to it, but he cannot choose less than $d$ such parents}.

\end{list}

However, our scheme would not work well for highly changing topologies that also do not fall under any of~\itemref{c:memb} or~\itemref{c:link}. Such an example are military ad-hoc wireless networks where the nodes are in constant rapid movement; this would not allow a child to discover his grandparents effectively.




\subsection{Extensions}\label{ref:extensions}


In this section, we describe how our protocol could be applied to other network coding scenarios.


First, note that we did not make any assumption about what a link or a node really is. A link can be a physical link, a chain of physical links, or even a subnetwork. For example, in a peer to peer network, a link can include an entire subnetwork via which some peers send data to a receiving peer. In this case, our protocol can be used to check that the receiving peer coded over all sender peers when he forwards the packets to some other peer. As another example, a link in a wired network may represent a connection, while a link in a wireless network may be the  ability to hear/communicate with another node or be an edge induced by the data  transmission graph. Moreover, a node can be a physical node (a router, a peer in a P2P network) or a subnetwork; in fact, a few nodes in our model can form one node for a certain system. Using these observations, we can express constraints of real-world networks:

\parhead{Multiple packets may be sent on some links}
Consider that parent $\PARENT_{i}$ has a capacity of $p$ packets on the link to node $\NODE$. In this case, in our protocol, $\PARENT_{i}$ will be represented as $p$ different nodes, each with a different public key. With this transformation, our protocol can be used unchanged.

\parhead{Broadcast links} 
Broadcast in wireless can be mapped to our model by having the parent have one link (the same link) to all his children (basically, viewing all children as one child), and our protocols can be applied unchanged.

\parhead{Multi-source network coding}
In the multi-source network coding case, intermediate nodes combine packets for different files from different sources, but each source operates independently and may not communicate with the others. In such work, the metadata of the packet is augmented with information about which source and which file identifiers the current packet contains. 

To support our protocols in the multi-source case, note that \one{} and \two{} depend on source information only when checking validity signatures. Moreover, our protocols are built modularly on top of a validity signature and do not depend on any particular scheme. This means that all we need is a multi-source validity signature and the rest of the algorithms will remain unchanged. Recent work~\cite{AgrawalBBF10} proposes such schemes: sources can send packets independently of each other, each packet contains a validity signature, and these signatures can be checked at each intermediate node by knowing the public keys of each of these sources.  Children will be able to check if their parents coded over the appropriate grandparents as before.

\parhead{Asynchronous networks and delay intolerant networks} 
A child may receive data from his parents at different times. For efficiency reasons, the child may have to code over the data that he received already and send the data forward, and not wait until a piece arrived from every parent. In this case, the child $\NODE$ can enforce the threshold verification above, thus checking that the packet from $\NODE$ is coded over at least a few parents. 





\parhead{Various levels of abstraction} 
Our protocol can be used at various levels of abstraction. For example, in peer-to-peer networks, nodes can perform:

\begin{itemize}[nolistsep]

  \item \emph{End-to-end check.} 
        A peer can check that the data from another peer is the result of coding over the data of all of certain sources, even if those sources communicated with the tested peer via other nodes or networks. 
        
  \item \emph{Individual node check.}
        A peer can check that the data from another peer is the result of coding over all of certain peers to which this peer should be connected to according to the Peer-to-Peer algorithm they run or whatever application they run.
\end{itemize}

A lot of P2P systems are taking advantage of smartphones nowadays. In \sectionref{sec:evaluation}, we show that our protocol is efficient even when run on a smart phone such as Android Nexus One.  




\section{Implementation and Evaluation}
\label{sec:implementation-and-evaluation}\label{sec:evaluation}



In this section, we evaluate the usefulness and the performance of our protocol.




\subsection{Simulation}
\label{sec:simulation}

We run a Python simulation to show that there is significant throughput loss due to Byzantine behavior not detected in previous work, but detected in our protocols. 
We examined three types of node behavior:
 \emph{(Mode 1)} Byzantine nodes choose coding coefficients such that their packet does not provide new information at their children;
\emph{(Mode 2)} Byzantine nodes simply forward one of the received packets (and do not code);
\emph{(Mode 3)} Byzantine nodes are forced to code with pseudorandom coefficients.
We can see that neither Mode $1$ nor Mode $2$ are detected by prior work on pollution schemes, but both are detected by our protocols. Mode $3$, which is the correct behavior, is enforced only by our protocols.

\begin{figure}[t!]
\centering
\includegraphics[scale=0.39]{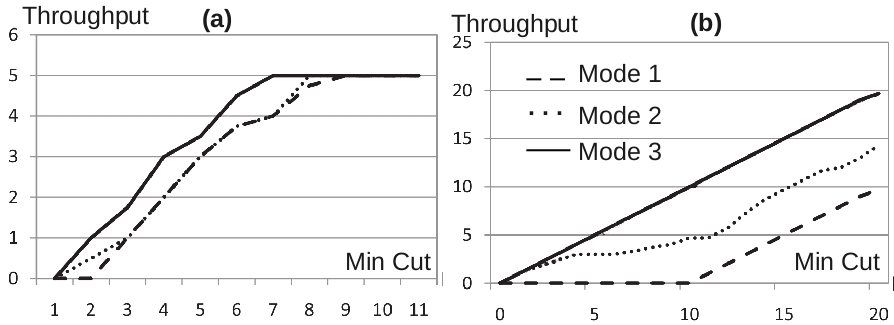}
\captionsetup{font=footnotesize}
\caption{(a) One and (b) ten Byzantine nodes on the mincut. \label{fig:graph1} \vspace{-0.7cm}}
\end{figure}



The simulation constructs a graph by assigning edges at random between nodes, but maintaining the given minimum cut. The Byzantine nodes are placed on the minimum cut. 
We ran the simulation for [$50$ nodes, $1000$ edges, $5$ packets sent from the source, min-cut up to $10$, 1 Byzantine node] and [$100$ nodes, $2000$ edges, $20$ packets send from the source, min-cut value up to $20$, $10$ Byzantine nodes]. \figureref{fig:graph1} shows the throughput (i.e., the degrees of freedom) at the sink plotted against the min-cut value.
We can see  that the throughput difference between Modes $1$/$2$ and Mode $3$ is significant. Moreover, when the min-cut value of the network is small (e.g., $5$), the throughput increase when using Mode $3$ can be as large as twice (see min-cut value of $3$ in \figureref{fig:graph1}(a)). 
In \figureref{fig:graph1} (b), we can see a more significant throughput difference. Mode $3$ has a throughput of about $10$ degrees of freedom more than Mode $1$ (which is $50\%$ of the data sent by the source) and about $5$ degrees of freedom more than Mode $2$ (which is $25\%$ of the data sent by the source). 





\subsection{Implementation}
\label{sec:implementation}

We implemented our protocol as a library (called \textit{SecureNetCode}) in C/C++ and Java, as well as embedded it into the Android platform. The C/C++ implementation is useful for lower level code that is meant to be fast: network routers, various wireless settings, and other C/C++ programs. The Java implementation is useful for higher-level programs such as P2P applications. We embedded the Java implementation in the Android platform and ran it on a Nexus One smartphone. The reason is that, with the growing popularity of smartphones, more P2P content distribution applications for smartphones are developed, some using network coding (\cite{socialtv}, \cite{mobile}).

Our library implementation is available at \url{www.mit.edu/~ralucap/netcode.html}\hspace{1mm}. It consists of the functions in protocols \one{} and \two{}. 
Our library in C/C++  consists of $290$ lines and the one in Java consists of $274$ lines including comments and white lines, but excluding standard, number theory or cryptographic libraries. To implement certain cryptographic operations on large numbers, we used NTL in C/C++ and BigInteger in Java. As cryptographic algorithms, we used OpenSSL DSA and SHA. The size of the validity signature used is $1024$-bit.

\parhead{Results}  
Except for the Android results which were run on a standard Nexus One smartphone, the rest of the results were run on a dual-core processor with $2.0$ GHz and $1$ GByte of RAM. There was observable variability in the results (especially for Nexus One), so we ran the experiments up to $100$ times to find an average time.

Note that we only evaluate the performance of our diversity scheme and do not evaluate the performance of any pollution signature protocol. The reason is that our protocol is not tied to any particular such scheme and uses it modularly. To enforce that nodes code with coefficients of one (\sectionref{sec:force-valid}), the most important step for throughput, we invoke the pollution scheme no more than it is invoked without our diversity checks. To enforce our full protocol with pseudorandom coefficients, during verification, each node computes one additional homomorphic operation of the integrity signature (per parent for \one{} and per challenge for \two{}), typically an exponentiation in a certain group: $\sig_{\so}(E_{\pa_i})^{\alpha^*_i}$. Fortunately, the coding coefficients are typically relatively small, e.g., 64 bits (even though the integrity signature allows them to be as large as $q$ as explained in \sectionref{sec:notation}). Note that the pollution signature verification, which is expensive, is not called additionally. 


In \tableref{table:results}, we present performance results of \one{} and \two{} using one challenge. We consider an integrity signature of size $1024$ bits and coding coefficients of size $64$ bits. 

\begin{table}[t!]
\centering
{\footnotesize
\begin{tabular}{@{}r|cc|cc|cc|@{}}
\cline{2-7}
& \multicolumn{2}{c|}{C/C++} & \multicolumn{2}{c|}{Java} & \multicolumn{2}{c|}{Android} \\
  & \one & \two &  \one &  \two & \one &  \two \\ 
\bottomrule
\multicolumn{1}{|r|}{1}  & 0.2/0.3 & 0.3/0.2 & 2.3/4.5 & 2.7/4.5  & 4.7/4.2  & 4.9/6.9   \\
\hline
\multicolumn{1}{|r|}{2}  & 0.2/0.6 & 0.3/0.2 & 2.3/9   & 2.7/4.6  & 4.7/7.6  & 5.1/7.1   \\
\hline
\multicolumn{1}{|r|}{3}  & 0.2/0.8 & 0.3/0.3 & 2.3/14  & 2.8/4.6  & 4.7/15.4 & 5.7/10.4  \\
\hline
\multicolumn{1}{|r|}{5}  & 0.2/1.4 & 0.3/0.3 & 2.3/23  & 2.8/4.7  & 4.7/24.4 & 6.7/10.5  \\
\hline
\multicolumn{1}{|r|}{7}  & 0.2/1.9 & 0.3/0.3 & 2.3/32  & 2.9/4.7  & 4.7/35.4 & 10.2/10.8 \\
\hline
\multicolumn{1}{|r|}{10} & 0.2/2.8 & 0.3/0.4 & 2.3/45  & 2.9/4.7  & 4.6/70.6 & 11.9/10.3 \\
\hline
\multicolumn{1}{|r|}{15} & 0.2/4.2 & 0.3/0.4 & 2.3/68  & 3.0/4.7  & 4.6/101  & 11.7/10.4 \\
\hline
\multicolumn{1}{|r|}{50} & 0.3/14  & 0.4/0.4 & 2.3/224 & 3.4/ 4.7 & 4.6/351  & 28.5/15.6 \\
\hline
\hline 
 \multicolumn{1}{|r|}{$+$}   & 0.95$|\REQUIREDSET|$  & 0.95 & 8.8$|\REQUIREDSET|$ & 8.8 & 15.4$|\REQUIREDSET|$ & 15.4 \\
\toprule
\end{tabular}
}
\captionsetup{font=footnotesize}
\caption{Performance results of \one{} and \two{} in milliseconds. The first $8$ rows with values show results for \one{} and \two{} when all coding coefficients are one (\sectionref{sec:force-valid}). The first column indicates the number of parents of a node. Each data cell in the rest of the columns consists of two values: transmission time and verification time. The last row shows the additional cost (only for verification) when adding pseudorandom coefficients (\sectionref{sec:force-diverse}) due to the homomorphic operation of the validity signature. \vspace{-0.7cm} }
\label{table:results}
\end{table}

We can see that, for verification, as we increase the number of parents, the overhead of \two{} increases very slowly (logarithmically) as compared to the linear performance of \one. The same happens to packet size, which we evaluate later in this section. Therefore, \textit{we recommend using \two{} for scenarios with more than three parents, and \one{} for cases with at most three parents.}  Alternatively, one could select a hybrid algorithm by performing $r > 1$ challenges from \two. The performance of \two{} grows linearly in the number of challenges so one can tune the probability of detection (see \sectionref{sec:protocol}) based on the desired tradeoff with performance overhead.

We can see that the C/C++ protocols impose modest overhead. For $10$ parents, which is a reasonably large value, the running time at a node to prepare for transmitting the data is $\approx0.25$ ms and the time to verify a packet's diversity $1.4$ ms in total for \two{}; for three parents, the time to verify diversity is $3.7$ ms for \one. All these values are \textit{independent of how large the packet payload is}. Let's compare this to the cost of a pollution scheme, for example~\cite{BonehFKW09}. In this scheme, the verification consists of two bilinear map computations and $m+n$ modular exponentiations, resulting in at least $100$ ms run time for verification in C using the PBC library for bilinear maps for each parent. For three parents, the relative overhead of \one{} is thus $<2\%$ and of \two{} is $<0.5\%$. Due to this low additional overhead, we believe that if one is already using a pollution scheme, one might as well also use our scheme in addition to provide diversity. 

The Java and Android implementations are slower because of the language and/or device limitations of the Nexus One. Nevertheless, we believe these implementations still perform well when used for higher level applications like P2P content distribution. 

\subsection{Packet Size}
\label{sec:packet-size}

For \one{}, the packet size increase in \one{} is $\INDEG{\NODE} \cdot (|\sigma_{\SOURCE}| + 320) + 320$ bits and the sum of packet increase and information sent during challenge phase in \two is $480 +  |\sigma_{\SOURCE}|  + 2 |\sigma_{\SOURCE}| \log(\INDEG{\NODE}) $ bits, where $\INDEG{\NODE}$ is the number of parents to code over. Recall that $|\sigma_{\SOURCE}|$ is the size of the validity signature, and depends on the validity scheme used. For instance, if \cite{BonehFKW09} is used, we have an increase in \one{} of  $480 \cdot \INDEG{\NODE}  + 320$ bits and in \two{} of $640 +  320 \cdot \log(\INDEG{\NODE})$ bits. As discussed in \sectionref{sec:protocol}, the packet size does not increase as the payload grows, so such overhead becomes insignificant when transmitting large files.

\section{Conclusions}
\label{sec:conclusions}

In this paper, we presented two novel protocols, \one{} and \two{}, for detecting whether a node coded correctly over all the packets received according to a random linear network coding algorithm. No previous work defends against such diversity attacks by Byzantine nodes. Our evaluation shows that our protocols are efficient and the overhead of both of our protocols does not grow with the size of the packet payload.

\bibliographystyle{alpha}
\bibliography{netcode}

\end{document}